\documentclass{eptcs}

\usepackage{amsmath}
\usepackage{amsfonts}
\usepackage{amssymb}
\usepackage{amsthm}
\usepackage{todonotes}
\usepackage[linesnumbered,boxed,commentsnumbered]{algorithm2e}
\usepackage{booktabs}
\usepackage{paralist}
\newcommand{\mdp}{M}
\newcommand{\mdptext}{\textnormal{\textsc{Mdp}}}
\newcommand{\lptext}{\textnormal{\textsc{Lp}}}

\newcommand{\mdpstext}{\textnormal{\textsc{Mdp}}s}
\newcommand{\ltltext}{\textnormal{\textsc{Ltl}}}
\newcommand{\mctext}{\textnormal{\textsc{Mc}}}
\newcommand{\prism}{\textnormal{\textsc{Prism}}}
\newcommand{\pctlstarttext}{\mbox{\textnormal{\textsc{Pctl}${}^{*}$}}}
\newcommand{\states}{\mathsf{S}}
\newcommand{\trans}{\mathsf{T}}

\newcommand{\state}{s}
\newcommand{\statep}{t}
\newcommand{\pstate}{s'}

\newcommand{\tran}{\mu}

\DeclareMathOperator{\enabledWord}{en}

\newcommand{\ist}{i}

\DeclareMathOperator{\lastWord}{last}
\newcommand{\last}[1]{\lastWord({#1})}
\newcommand{\pat}{\rho}

\DeclareMathOperator{\pathsWord}{Paths}

\newcommand{\sched}{\eta}

\newcommand{\plen}{n}
\newcommand{\enabledt}[1]{\enabledWord({#1})}

\newcommand{\pathsched}[2]{\pathsWord({#2},{#1})}
\newcommand{\prr}[2]{\pr^{#1}({#2})}

\newcommand{\prrmi}[4]{\pr^{{#4},{#1}}_{#3}({#2})}

\newcommand{\prrhmi}[4]{\pr^{{#4},{#1}}_{#3}(\reach{#2})}
\newcommand{\prrhminos}[3]{\pr^{{#3}}_{#2}(\reach{#1})}
\DeclareMathOperator{\distWord}{Dist}
\newcommand{\ddist}[1]{\distWord({#1})}
\newcommand{\extension}[1]{{#1}^{\uparrow}}

\DeclareMathOperator{\pr}{Pr}

\newcommand{\prismtext}{\textsc{Prism}}

\newcommand{\goesto}[3]{{#1} \xrightarrow{#2} {#3}}

\DeclareMathOperator{\lenWord}{len}
\newcommand{\pathlen}[1]{\lenWord({#1})}

\theoremstyle{definition}
\newtheorem{definition}{Definition}
\theoremstyle{plain}
\newtheorem{lemma}{Lemma}
\newtheorem{theorem}{Theorem}

\DeclareMathOperator{\reachWord}{reach}
\newcommand{\reach}[1]{\reachWord({#1})}

\newcommand{\tsts}{U}

\newcommand{\alltran}{m}
\newcommand{\allstates}{n}
\newcommand{\bi}{k}
\newcommand{\basis}{B}
\newcommand{\nbasis}{N}

\newcommand{\cannotreach}{\states^{\sup 0}}
\newcommand{\canavoidreach}{\states^{\inf 0}}
\newcommand{\maybe}{\states^{?}}
\newcommand{\card}[1]{|{#1}|}

\newcommand{\srestrict}[2]{{#1}\downarrow{#2}}

\newcommand{\apttext}{apt}

\newcommand{\lpproblem}{\mathcal{L}}
\DeclareMathOperator{\stateWord}{state}
\newcommand{\tstate}[1]{\stateWord({#1})}
\newcommand{\bsched}[1]{\basis_{#1}}

\theoremstyle{plain}
\newtheorem{notation}{Notation}
\newtheorem{assumption}{Assumption}

\newcommand{\pmatrixsched}[1]{(A\!\downarrow\!{#1})}
\newcommand{\matrixsched}[1]{C^{#1}}
\newcommand{\abs}[1]{|{#1}|}

\newcommand{\glpk}{\texttt{glpk}}

\begin{document}

\newcommand{\titlerunning}{Efficient computation of exact solutions for quantitative model checking}
\newcommand{\authorrunning}{Sergio Giro}

\title{Efficient computation of exact solutions for quantitative model checking}

\author{Sergio Giro\institute{Department of Computer Science, University of Oxford, Oxford, OX1
3QD, UK\thanks{This work was supported by DARPA and the Air Force Research
Laboratory under contract FA8650-10-C-7077 (PRISMATIC). Any opinions, findings,
conclusions, or recommendations expressed in this material are those
of the authors and do not necessarily reflect the views of AFRL or
the U.S. Government.}}}

\maketitle

\renewcommand{\sup}{\max}
\renewcommand{\inf}{\min}

\begin{abstract}
Quantitative model checkers for Markov Decision Processes typically use finite-precision arithmetic.
If all the coefficients in the process are rational numbers, then the model checking results are rational,
and so they can be computed exactly. However, exact techniques are generally too expensive or limited in scalability.
In this paper we propose a method for obtaining exact results starting
from an approximated solution in finite-precision arithmetic. The input of the
method is a description of a scheduler, which can be obtained by a model checker
using finite precision. Given a scheduler, we show how to obtain a corresponding
basis in a linear-programming problem, in such a way that the basis is optimal whenever the scheduler attains the worst-case probability. This correspondence is
already
known for discounted MDPs, we show how to apply it in the undiscounted case provided
that some preprocessing is done. Using the correspondence,
the linear-programming problem can be solved in exact arithmetic starting
from the basis obtained. As a consequence, the method finds the worst-case 
probability even if the scheduler provided by the model checker was not optimal.
In our experiments, the calculation of exact solutions from a
candidate scheduler is significantly faster than the calculation using
the simplex method under exact arithmetic starting from a default basis.
\end{abstract}

\section{Introduction}

Model checking of Markov Decision Processes (\mdpstext)
has been proven to be a useful tool to verify and evaluate systems with
both probabilistic and non-deterministic choices. Given a model of the system
under consideration and a qualitative property concerning probabilities, such as
``the system fails to deliver a message with probability at most $0.05$'', a
model checker deduces whether the property holds or not for the model. As different resolutions
of the non-deterministic choices lead to different probability values, verification
techniques for \mdpstext\ rely on the concept of \emph{schedulers} (also called policies,
or adversaries), which are defined as functions choosing an option for each of the paths
of an \mdptext. Model-checking algorithms for \mdpstext\ proceed by
reducing the model-checking problem to that of finding the
maximum (or minimum) probability to reach a set of states
under all schedulers~\cite{DBLP:conf/fsttcs/BiancoA95}.

Different techniques for calculating these extremal probabilities
exist: for an up-to-date tutorial, see~\cite{DBLP:conf/sfm/ForejtKNP11}. Some
of them (for instance, value iteration)
are approximate in nature.
If all the coefficients in the process are rational numbers, then the model checking results are rational,
and so they can be computed exactly. However, exact techniques are generally too expensive or limited in scalability.
Linear programming (\lptext)
can be used to obtain exact solutions, but in order to achieve reasonable efficiency
it is often carried out using finite-precision, and so the results are always
approximations. (We performed some experiments showing how
costly it is to compute exact probabilities using \lptext\ without our method.)
In addition, the native operators in programming languages like Java
have finite precision: the extension to exact arithmetic involves significant
reworking of the existing code.

We propose a method for computing exact solutions. Given any approximative
algorithm being able to provide a description of a scheduler, our method
shows how to extend the algorithm in order to get exact solutions. The method
exploits the well-known correspondence between model-checking problems
and linear programming problems~\cite{DBLP:conf/fsttcs/BiancoA95}, which allows
to compute worst-case probabilities by computing optimal solutions for \lptext\ problems.
We do not know of any similar approach to get exact values. This might be due to
the fact that, from a purely theoretical point of view, the problem is not very
interesting, as exact methods exists and they are theoretically efficient: the problem
is that, in practice, exact arithmetic introduces significant overhead.

The simplex algorithm~\cite{book:applied-mathematical-programming}
for linear programming works by iterating over different \emph{bases}, which are
submatrices of the matrix associated to the \lptext\ problem. Each basis defines
a \emph{solution}, that is, a valuation on the variables of the problem.
The simplex method stops when the basis yields a solution with certain properties,
more precisely, a so-called \emph{feasible} and \emph{dual feasible} solution. By algebraic
properties, such a solution is guaranteed to be optimal.

The core of our method is the interpretation of the scheduler as a basis for the
linear programming problem. Given a scheduler complying with certain natural
conditions, a basis corresponding to the scheduler can be used as a starting point for
the simplex algorithm. We show that, if the scheduler is optimal, then the solution
associated to the corresponding basis is feasible and dual feasible, and so a
simplex solver provided with this basis needs only to check dual feasibility and
compute the solution corresponding to the basis. As our experiments show, these
computations can be done in
exact arithmetic without a huge impact in the overall model-checking time. In fact,
using the dual variant of the simplex method,
the time to obtain the exact solution is less than the time spent
by value iteration.
 If the scheduler is not
optimal, the solver starts the iterations from the basis. This is useful
for two reasons: we can let the simplex solver finish in order to get the exact 
solution; or, once we know that we are not getting the optimal solution, we can
perform some tuning in the model checker as, for instance, reduce the convergence
threshold (we also show a case in which the optimal scheduler cannot be found with thresholds within the 64-bit IEEE~754 floating point precision).

The correspondence between schedulers and bases is already known for discounted \mdpstext\ (see, for instance~\cite{depen-1963:probaprodu:journal:98}). We show the correspondence
for the undiscounted case in case some states of the system are eliminated in preprocessing
steps. The preprocessing steps we consider are usual in model
checking~\cite{DBLP:conf/sfm/ForejtKNP11}: given a set of target states, one of the
preprocessing algorithms removes the states that cannot reach the target, while the other
one removes the states that can avoid reaching the target. These are qualitative algorithms
based on graphs that do not perform any arithmetical operations.

The next section introduces the preliminary concepts we need along the paper.
Section~\ref{sec:method} presents our method and the proof of correctness.
The experiments are shown in Section~\ref{sec:exp-results}. The last section
discusses related results concerning complexity and policy iteration.

\section{Preliminaries}

We introduce the definitions and known-results used throughout the paper, concerning both
Markov decision process and linear programming.

\subsection{Markov decision processes}

\begin{definition}
\label{def:mdp}
Let $\ddist{A}$ denote the set of discrete probability distributions
over the set $A$. A Markov Decision Process (\mdptext) $\mdp$ is a pair
$(\states, \trans)$ where $\states$ is a finite
set of \emph{states} and $\trans \subseteq \states \times \ddist{\states}$ is
a set of \emph{transitions}~\footnote{Defining transitions as pairs helps to deal
with the case in which the same distribution is enabled
in several states}. Given $\tran = (\state, d) \in \trans$, the value $d(\statep)$
is the probability of making a transition to $\statep$ from $\state$ using $\tran$.
We write 
$\tran(\statep)$ instead of $d(\statep)$, and write $\tstate{\tran}$ for $\state$.
We define the set $\enabledt{\state}$ as the set of
all transitions $\tran$ with $\tstate{\tran} = \state$. For simplicity, we make the
usual assumption that every state has at least one enabled transition:
$\enabledt{\state} \not= \emptyset$ for all $\state \in \states$.
\end{definition}

We write $\goesto{\state}{\tran}{\statep}$ to denote
$\tran \in \enabledt{\state} \land \tran(\statep) > 0$. 
A path in an \mdptext\ is a (possibly infinite) sequence
$\pat = \state^{0}.\tran^{1}.\state^{1}.\cdots.\allowbreak
							.\tran^{\plen}.\state^{\plen}$,
where
$\tran^{\ist} \in \enabledt{\state^{\ist-1}}$ and
$\tran^{\ist}(\state^{\ist}) > 0$ for all
$\ist$.
If $\pat$ is finite, the last state of $\pat$ is denoted by
$\last{\pat}$, and the length is denoted by $\pathlen{\pat}$
(a path having a single state has length $0$).
Given a set of states $\tsts$, we define $\reachWord(\tsts)$ to be
the set of all infinite paths
$\pat = \state^{0}.\tran^{1}.\state^{1}.\cdots$ such
that $s^{i} \in \tsts$ for some $i$.


The semantics of \mdpstext\ is given by schedulers. A scheduler $\sched$ for
an \mdptext\ $\mdp$ is a function
$\sched\colon \states \to \trans$ such that
$\sched(\state) \in \enabledt{\state}$ for all $\state$. In words, the scheduler
chooses an enabled transition based on the current
state.
For all schedulers $\sched$, $\statep \in \states$,
the set $\pathsched{\sched}{\statep}$ contains all the paths
$\state^{0}.\tran^{1}.\state^{1}.\cdots\allowbreak
							.\tran^{\plen}.\state^{\plen}$ such that
$\state^{0} = \statep$, $\tran^{\ist} = \sched(\state^{\ist-1})$
and $\goesto{\state^{\ist-1}}{\tran^{\ist}}{\state^{\ist}}$
for all $\ist$.
The reader familiar with \mdpstext\
might note that we are restricting to Markovian non-randomized schedulers
(that is, they map states to transitions, instead of the more general schedulers
mapping paths to distributions on transitions).
As explained later on, these schedulers suffice for our purposes.

The probability
$\prrmi{\sched}{\pat}{\mdp}{\statep}$ of the path $\pat$ under
$\sched$ starting from $\statep$ is
$\prod_{\ist=1}^{\pathlen{\pat}} \tran^{\ist}(\state^{\ist})$
if $\pat \in \pathsched{\sched}{\statep}$. If $\pat \not\in \pathsched{\sched}{\statep}$,
then the probability is $0$. We often omit the subindices $\mdp$ and/or $\statep$
if they are clear from the context.

We are interested on the probability of (sets of) infinite paths.
Given a finite path $\pat$, the probability of the set $\extension{\pat}$
comprising all the infinite paths that have $\pat$ as a prefix is defined by
$\prr{\sched}{\extension{\pat}} = \prr{\sched}{\pat}$. In the
usual way (that is, by resorting to the Carath\'eodory extension theorem)
it can be shown that the definition on the sets of the form $\extension{\pat}$
can be extended to $\sigma$-algebra generated by the sets $\extension{\pat}$.


The verification of \pctlstarttext~\cite{DBLP:conf/fsttcs/BiancoA95} and
$\omega$-regular formulae~\cite{DBLP:conf/icalp/CourcoubetisY90}
(for example \ltltext) can be reduced
to the problem of calculating $\sup_{\state,\sched} \prrhmi{\sched}{\tsts}{\mdp'}{\state}$
(or $\inf_{\state,\sched} \prrhmi{\sched}{\tsts}{\mdp'}{\state}$) for \mdpstext\
$\mdp'$, states $\state$ and sets $\tsts$ obtained from the formula.

In consequence, in the rest of the paper we concentrate on the following
problems.

\begin{definition}
Given an \mdptext\ $\mdp$, an initial state $\state$ and set of \emph{target}
states $U$, a \emph{reachability problem} consists of computing
$\sup_{\state,\sched} \prrhmi{\sched}{\tsts}{\mdp}{\state}$ (or
$\inf_{\sched} \prrhmi{\sched}{\tsts}{\mdp}{\state}$).
\end{definition}

From classic results in \mdptext\ theory (for these results applied to
model checking see, for instance, \cite[Chapter 3]{STAN//CS-TR-98-1601})
there exists a scheduler $\sched^{*}$ such that
\begin{equation}
\label{eq:optimal-sched-sup}
 \sched^{*} = \arg \sup_{\sched} \prrhmi{\sched}{\tsts}{\mdp}{\state}
\end{equation}
for all $\state \in \states$. That is, $\sched^{*}$ attains the maximum probability for
all states.

An analogous result holds for the the case of minimum probabilities.
There exists $\sched^{*}$ such that
\begin{equation}
\label{eq:optimal-sched-inf}
 \sched^{*} = \arg \inf_{\sched} \prrhmi{\sched}{\tsts}{\mdp}{\state}
\end{equation}
for all $\state \in \states$.

Even in a more general setting allowing for non-Markovian and randomized
schedulers, it can be proven that we can assume $\sched^{*}$ to be Markovian and
non-randomized. The existence of $\sched^{*}$ justifies our restriction
to Markovian and non-randomized schedulers.

\paragraph{Markov chains}
A Markov chain (\mctext) is an \mdptext\ such that $\card{\enabledt{\state}} = 1$ for all
$\state \in \states$. 
Note that a Markov chain has exactly one scheduler,
namely the one that chooses the only transition enabled in each
state.
Hence, for Markov chains, we often disregard the scheduler and denote the
probability of reaching $\tsts$ as $\prrhminos{\tsts}{\mdp}{\state}$.

\begin{definition}
\label{def:restrict-mc}
Given an \mdptext\ $\mdp  = (\states,\trans)$ and a scheduler
$\sched$, we define the Markov chain
$\srestrict{\mdp}{\sched} = (\states, \trans')$ where
$\tran \in \trans'$ iff $\sched(\tstate{\tran}) = \tran$.
\end{definition}

A simple application of the definitions yields
\begin{equation}
\label{eq:rest-mc-eq-probs}
 \prrhmi{\sched}{\tsts}{\mdp}{\state}
          = \prrhminos{\tsts}{\srestrict{\mdp}{\sched}}{\state} \; .
\end{equation}

\subsection{Linear programming}
\label{subsec:lin-prog}
\renewcommand{\vec}[1]{\mathbf{#1}}

We use a particular canonical form of linear programs suitable for our needs.
It is based on \cite[Appendix B]{book:applied-mathematical-programming}, which
is also a good reference for all the concepts and results given in this subsection.

A linear programming problem consists in computing
\begin{equation}
\label{eq:lp-matrix-form}
\min_{\vec{x}} \{ \vec{c}\vec{x} \mid A\vec{x} = \vec{b} \land \vec{x} \geq \vec{0} \} \; ,
\end{equation}
given a \emph{constraint matrix} $A$, a \emph{constraint vector} $\vec{b}$
and a \emph{cost vector} $\vec{c}$. In the following, we assume that $A$ has $\alltran$ rows
and $\alltran + \allstates$ columns, for some $\alltran > 0$ and $\allstates \geq 0$. Hence,
$\vec{c}$ is a row vector with $\alltran + \allstates$ components, and $\vec{b}$ is a column
vector with $\alltran$ components.

A \emph{solution} is any vector $\vec{x}$ of size $\alltran+\allstates$. The
$i$-th component of $\vec{x}$ is denoted by $x_{i}$. We say that
that a solution is \emph{feasible} if $A\vec{x} = \vec{b}$ and
$\vec{x} \geq \vec{0}$;
it is optimal if is feasible and $\vec{c}\vec{x}$ is minimum over all feasible
$\vec{x}$. A problem is \emph{feasible} if it has a feasible solution, and
\emph{bounded} if it has an optimal solution.
A non-singular $\alltran \times \alltran$ submatrix of $A$ is called a \emph{basis}.
We overload the letter $\basis$ to denote both the basis and the set of indices of the
corresponding columns in $A$. A variable $x_{\bi}$ is basic if $\bi \in \basis$.
Note that, given our assumptions on the dimension
of the constraint matrix, for all bases there are $\alltran$ basic variables
and $\allstates$ non-basic variables.
Given a basis $\basis$, and any vector $\vec{t}$, let $\vec{t}^{\basis}$
be the subvector of $\vec{t}$ having only the components in $\basis$. 
When $\basis$ is clear from the context, we use $\nbasis$
to denote the set of columns \emph{not in} $\basis$, and
use $\vec{t}^{\nbasis}$ accordingly. For a matrix
$A$, let $A^{\nbasis}$ be submatrix of $A$ having only the columns that are not in $\basis$.
The solution $\vec{x}$ \emph{induced} by the basis $\basis$ is defined as
$x_{\bi} = 0$ for all $\bi \not\in \basis$, while the values
for $\bi \in \basis$ are given by the vectorial equation
$x^{\basis} = \basis^{-1} \vec{b}$.
A solution $\vec{x}$ is \emph{basic} if there is a basis that induces $\vec{x}$.
Given $\basis$ and $\bi \in \nbasis$, the \emph{reduced cost} $\overline{c_{\bi}}$
of a variable $x_{\bi}$ is defined as
$c_{\bi} - \vec{c}^{\basis}\basis^{-1}A_{\bi}$,
where $A_{\bi}$ is the $\bi$-th column of $A$. A solution is \emph{dual feasible} if
it correspond to a basis such that $\overline{c_{\bi}} \geq 0$ for all
$\bi \in \nbasis$.

In our proofs we make use of the following lemma, which is particular
to our canonical form.
\begin{lemma}
\label{lemma:basic-satisfies-eq}
\[ A\vec{x} = \vec{b} \qquad \mbox{if $\vec{x}$ is basic} \;. \]
\end{lemma}
\begin{proof}
By splitting $A$ into basic and non-basic columns we get
$
 A \vec{x} = \basis \vec{x}^{\basis} + A^{\nbasis} \vec{x}^{\nbasis}   
             = \basis \basis^{-1} \vec{b} + A^{\nbasis} \vec{0}
             = I \vec{b} = \vec{b} \; . $
(Note that $\vec{x}$ might not be feasible as it could be
$\vec{x} \not\geq \vec{0}$.)
\end{proof}

Correctness of the simplex method relies on the following well-known facts about \lptext\ problems:
\begin{compactitem}
\item{Every solution that is both feasible and dual feasible is optimal}
\item{If there exists an optimal solution, then there exists a \emph{basic} solution
that is feasible and dual feasible (and hence optimal)}
\end{compactitem}

As the problems we deal 
with are ensured to be bounded and feasible, we assume that there exists an optimal solution.
In this context, the simplex algorithm explores different bases until it finds a basis whose
corresponding solution is feasible and dual feasible.

In several implementations of the algorithm
the starting basis can be specified (when it is not, a default one is used). The initial basis
does not need to be feasible nor dual feasible. In case the starting basis complies with both
feasibilities, the simplex algorithm finishes after checking that these feasibilities are met,
without any further exploration. In Subsection~\ref{subsec:subsec-lp-mdp}, we show how
reachability problems correspond to \lptext\ problems.
In Section~\ref{sec:method} we show that,
under a certain assumption on the model
checker (Assumption~\ref{ass:apt}), a basis can be obtained
from the scheduler provided by the model checker. In particular, optimal schedulers 
yield feasible bases (Theorem~\ref{thm:optimal-is-feasible}). Under
our assumption, all the bases obtained from schedulers are dual feasible
(Theorem~\ref{thm:apt-is-dual-feasible}).

Among the different variants of the simplex method, in our experiments
(Section~\ref{sec:exp-results}) we use the \emph{dual} simplex, which first looks for a
dual-feasible basis (in the so-called \emph{first phase}) and next tries to find a feasible
one while keeping dual feasibility (in the \emph{second phase}). This is appropriate in our
case since, under our assumptions, the first phase is not needed (as
formalized in Theorem~\ref{thm:apt-is-dual-feasible}). In contrast to the dual simplex,
the \emph{primal} simplex (or, simply, simplex) looks for
a feasible basis in the first phase. As a consequence, if iterations are required (according
to our results in Section~\ref{sec:method}, this is case in which the model checker fails to
provide the optimal scheduler), then the primal simplex performs both phases. However, both
variants can be used and, as our experiments show, the starting basis obtained from the
scheduler is useful to save iterations. In the few cases in which \prism\ did not provide
the optimal schedulers, the dual simplex required less
iterations than the primal one; both of them perform far better when starting from a basis
corresponding to a near-optimal scheduler than when starting from the default basis
(see Section~\ref{sec:exp-results}).

\subsection{Linear programming for Markov decision processes}
\label{subsec:subsec-lp-mdp}
Linear programming can be used to compute optimal probabilities for some of the states
in the system. The set of states whose maximum (minimum, resp.) probability is $0$ is first
calculated using graph-based techniques~\cite[Sec. 4.1]{DBLP:conf/sfm/ForejtKNP11}.
This qualitative calculation is often considered as a preprocessing step
before the proper quantitative model checking.
 Given a set of target states $\tsts$, let $\cannotreach$
be the set of states $\states$ such that
$\sup_{\sched} \prrhmi{\sched}{\tsts}{\mdp}{\state} = 0$.
Similarly, let $\canavoidreach$ be the set of states such that
$\inf_{\sched} \prrhmi{\sched}{\tsts}{\mdp}{\state} = 0$. When focusing on
maximum probabilities, we write the set
$\states \setminus (\cannotreach \cup \tsts)$
as $\maybe$ (called the set of \emph{maybe} states), while for
minimum probabilities $\maybe$ is
$\states \setminus (\canavoidreach \cup \tsts)$.

The maximum probabilities for $\state \in \cannotreach$ are $0$ by definition
of $\cannotreach$. For $\state \in \tsts$ the probabilities are $1$, since when starting
from a state in $\tsts$, the set $\tsts$ is reached in the initial state,
regardless of the scheduler. The minimum probabilities for $\state \canavoidreach$ are $0$ by definition
of $\canavoidreach$, and the probabilities for $\state \in \tsts$ are again $1$.
Next we show how to obtain the probabilities for the states in $\maybe$,
thus covering all the states in the system.

In order to avoid order issues, we assume that
the states are $\maybe = \state_{1}, \cdots, \state_{\allstates}$ and
the transitions are:
\begin{equation}
\label{eq:def-trans}
\trans = \tran_{1}, \cdots, \tran_{\alltran}
\end{equation}
in such a way that if $s_{i} = \tstate{\tran_{j}}$, $s_{i'} = \tstate{\tran_{j'}}$
and $i < i'$, then $j < j'$ (from
Def.~\ref{def:mdp}, recall that $\tstate{\tran_{i}}$ is the state in which $\tran_{i}$
is enabled). Roughly speaking, the transitions are ordered with respect to
the states in which they are enabled. From now on, we use this orderings
consistently throughout the paper.

In the following theorem, the matrix $A|I$ associated to a reachability problem
$\sup \prrhmi{\sched}{\tsts}{\mdp}{\state}$ is a $\alltran \times (\allstates + \alltran)$
matrix whose last $\alltran$ columns form the identity matrix.
We define of $A_{i,j}$ for the column $j \leq \allstates$ as:
$A_{i,j} = \tran_{i}(\state_{j})$ if $\state_{j} \not= \tstate{\tran_{i}}$,
or $A_{i,j} = \tran_{i}(\state_{j})-1$ if $\state_{j} = \tstate{\tran_{i}}$.
The vector $\vec{b}$ is defined as
$b_{i} = -\sum_{\state \in \tsts} \tran_{i}(\state)$.

\begin{theorem}
\label{thm:lp-problems}
For all states $\state_{i} \in \maybe$, the value
$\sup_{\sched} \prrhmi{\sched}{\tsts}{\mdp}{\state_{i}}$ is the
value of the variable $x_{i}$ in an optimal solution of the following \lptext\ problem:
\begin{equation}
\label{eq:sup-matrix-form}
\begin{array}{rl}
\min & (\overbrace{1,\cdots,1}^{\allstates},
                          \overbrace{0,\cdots,0}^{\alltran})\vec{x} \\
                             & (A \mid I^{\alltran\times\alltran})
                                        \vec{x} = \vec{b} \\
                              & \vec{x} \geq \vec{0} \; .
\end{array}
\end{equation}

Analogously, the value $\inf_{\sched} \prrhmi{\sched}{\tsts}{\mdp}{\state_{i}}$ is the
value of the variable $x_{i}$ in an optimal solution of the following \lptext\ problem.
\begin{equation}
\label{eq:inf-matrix-form}
\begin{array}{rl}
\min & -(\overbrace{1,\cdots,1}^{\allstates},
                          \overbrace{0,\cdots,0}^{\alltran})\vec{x} \\
                             & (-A \mid I^{\alltran\times\alltran})
                                        \vec{x} = -\vec{b} \\
                              & \vec{x} \geq \vec{0} \; .
\end{array}
\end{equation}
(Note that, in the constraint, the matrix $A$ is negated, while $I$ is not.)
\end{theorem}

This theorem is just the well-known correspondence between reachability problems
and \lptext\ problems~\cite{puterman},\cite[Section 4.2]{DBLP:conf/sfm/ForejtKNP11},
written in our \lptext\ setting.

The variables that multiply the columns in the identity matrix are called \emph{slack}
variables in the \lptext\ literature. They are also the variables $x_{\tran}$ in the following
notation.

\begin{notation}
\label{not:states-tran-as-indices}
From now on, we identify each column $1 \leq j \leq \allstates$ of $(A|I)$
with the state $\state_{j}$, and each column $\allstates < j \leq \allstates + \alltran$
with the transition $\tran_{j}$. Each row $i$ is identified with $\tran_{i}$.
In consequence, we write $A_{\tran,\state}$ for the elements of the matrix,
and $x_{\state}$ or $x_{\tran}$ for the components of the solution $\vec{x}$.
\end{notation}

\section{A method for exact solutions}
\label{sec:method}

Our method serves as a complement to a model checker being able to:
\begin{compactitem}
\item{calculate the set $\maybe$, and}
\item{give a description of a scheduler, that the model checker \emph{considers}
optimal based on finite precision calculations}
\end{compactitem}
We only require a weak ``optimality'' condition on the scheduler returned
by the model checker, which we refer to as \emph{\apttext}: we say that a
scheduler $\sched$ is \apttext\ iff $\prrhmi{\sched}{\tsts}{\mdp}{\state} > 0$
for all $\state \in \maybe$. In order words, we only require the scheduler to reach
$\tsts$ for all states that can reach it (no matter with which probability). In the
case of minimum probabilities, every scheduler is \apttext, since if we have
$\prrhmi{\sched}{\tsts}{\mdp}{\state} = 0$ for some $\sched$, then
$\state \not\in \maybe$ (by definition of $\canavoidreach$). 
For the case of the maximum, the existence of an \apttext\ scheduler follows from the
definition of $\maybe$, the scheduler $\sched^{*}$ in~(\ref{eq:optimal-sched-sup})
being a suitable witness.

\begin{assumption}
\label{ass:apt}
We assume that the model checker is able to provide an \apttext\ scheduler, in the sense
that our method is not guaranteed to return a value in case the scheduler is not \apttext.
\end{assumption}

\newcommand{\ReachAnalysis}{\texttt{reach\_analysis}}
\newcommand{\ConstructProblem}{\texttt{construct\_problem}}
\newcommand{\ConstructBasis}{\texttt{construct\_basis}}
\newcommand{\StartSimplexSolver}{\texttt{start\_simplex\_solver}}
\newcommand{\Error}{\texttt{error}}

Our method is described in the Algorithm~\ref{alg:main-algorithm}.
\begin{algorithm}
\SetKwInOut{Input}{input}\SetKwInOut{Output}{output}
\Input{An \mdptext\ $\mdp$ and a set of states $\tsts$}
\Output{$\vec{x}$ such that $x_{\state} = \sup_{\sched} \prrhmi{\sched}{\tsts}{\mdp}{\state}$
				~~~~~~($\inf_{\sched} \prrhmi{\sched}{\tsts}{\mdp}{\state}$, resp.)
				for all $\state \in \maybe$}
				
 \BlankLine				
 \tcp{Use model checker to get the set $\maybe$ and a scheduler}
 \nllabel{line:get-scheduler}$(\maybe, \sched) \leftarrow$ \ReachAnalysis($\mdp$, $\tsts$)\;
 $\lpproblem \leftarrow$ \ConstructProblem($\mdp$, $\maybe$)\;
 $\bsched{\sched} \leftarrow$ \ConstructBasis($\lpproblem$, $\sched$)\;
 \StartSimplexSolver($\lpproblem$, $\bsched{\sched}$) \;
 \uIf{the exact simplex solver finishes in one iteration}{
      \Return $\arg \min_{\vec{x}} \lpproblem$, obtained from the solver\;
 }
 \uElseIf(\tcp*[f]{$\eta$ is not optimal}){the solver performs several iterations}{
      \Return $\arg \min_{\vec{x}} \lpproblem$, obtained from the solver once it finishes\;
      \tcp*[h]{Or interrupt the solver and change the model checker parameters}
 }
 \nllabel{line:basis-singular}\ElseIf{the solver reports that the basis is singular} {
      \tcp{For the minimum, this case cannot happen}
      \Error\ $\eta$ is not \apttext\;
 }
\caption{\label{alg:main-algorithm}Method to get exact solutions}
\end{algorithm}
The function \ConstructProblem constructs the \lptext\ problems (\ref{eq:sup-matrix-form}) and
(\ref{eq:inf-matrix-form}).
Given $\sched$, the basis $\bsched{\sched}$ obtained by
\ConstructBasis is defined as
\begin{equation}
\label{eq:def-basis}
 \state \in \bsched{\sched}, \quad \mbox{for all $\state \in \maybe$}
             \qquad \qquad x_{\tran} \in \bsched{\sched}
					\Longleftrightarrow \sched(\tstate{\tran}) \not= \tran \; .
\end{equation}
Roughly speaking, the basis contains all states, and all the transitions that are
\emph{not} chosen by $\sched$. Sometimes (particularly in the proof of Theorem~\ref{thm:apt-is-dual-feasible}) we write $\bsched{\mdp',\sched}$ to
make it clear that the basis belongs to an 
\mdptext\ $\mdp'$.

The rest of this section is devoted to prove the correctness of the algorithm,
in the sense made precise by the following theorem (which is proven later).

\begin{theorem}
\label{thm:alg-correct}
If the algorithm returns a value, then the value corresponds to the \textbf{output}
specification. Moreover, if the scheduler $\sched$ provided by the model checker is
\apttext, then the matrix defined by~\textup{(\ref{eq:def-basis})} is a basis, and the
algorithm returns optimum values from the \lptext\ solver. If the scheduler provided by the
model checker is optimum as in~\textup{(\ref{eq:optimal-sched-sup})}, then the basis
in~\textup{(\ref{eq:def-basis})} is both feasible and dual feasible.
\end{theorem}

Recall from Subsection~\ref{subsec:lin-prog} that the simplex algorithm stops as soon as
it finds a solution that is feasible and dual feasible. Hence, the fact that
an optimal scheduler yields a basic, feasible and dual feasible solution causes the
simplex solver to stop as soon as the feasibility checks are finished.

The rest of this section is devoted to prove Theorem~\ref{thm:alg-correct}.
In our proofs we resort to the
following definitions and lemmata. The first definition uses indices as explained in
Notation~\ref{not:states-tran-as-indices}.

\begin{definition}
\label{def:submatrices}
Given a scheduler $\sched$, we write the set of transitions complying with
$\sched(\tstate{\tran}) = \tran$ as $\trans_{\sched} =
   \{ \tran^{1}, \cdots, \tran^{\allstates} \}$, and we assume that this ordering respects
   the ordering in~(\ref{eq:def-trans}).
We define $\matrixsched{\sched}$ to be the
$\allstates \times \allstates$ matrix whose elements are as
$\matrixsched{\sched}_{i,j} = \tran^{i}(\state_{j})$.
Consider the matrix $A$ in~(\ref{eq:sup-matrix-form}).
We define $\pmatrixsched{\sched}$ to be the $\allstates \times \allstates$ submatrix of $A$
comprising all the rows $\tran\in \trans_{\sched}$ and the columns $\state$ for all
$\state \in \maybe$.
\end{definition}

\begin{lemma}
\label{lemma:matrix-are-chosen}
The transitions
$\tran^{i} \in \trans_{\sched}$ comply with
$\tstate{\tran^{i}} = \state_{i} \:$ for all $\state_{i} \in \maybe$.
In consequence, $\sched(\state_{i}) = \tran^{i}$.
\end{lemma}
\begin{proof}
Note that since the order in $\trans_{\sched}$ respects the order in~(\ref{eq:def-trans}),
we have that the sequence
$\tstate{\tran^{1}}, \cdots, \allowbreak \tstate{\tran^{\allstates}}$ is a sequence of states
$\state_{j_{1}}, \cdots, \state_{j_{\allstates}}$ with $j_{1} \leq \cdots \leq j_{\allstates}$.
Since there are $\allstates$ states, and for each state $\state$ we have exactly one transition
$\tran$ such that $\sched(\state) = \tran$, it must be
$s_{j_{1}} = s_{1}, \cdots, s_{j_{\allstates}} = s_{\allstates}$. This implies 
$\tstate{\tran^{i}} = \state_{j_{i}} = \state_{i}$ as desired.
Using this equality and $\tran^{i} \in \trans_{\sched}$ we have
$\sched(\state_{i}) = \sched(\tstate{\tran^{i}}) = \tran^{i}$.
\end{proof}

\begin{lemma}
\label{lemma:pmatrix-matrix}%
For all $\sched$, we have $\pmatrixsched{\sched} = \matrixsched{\sched} - I$.
\end{lemma}
\begin{proof}
By definition of $\pmatrixsched{\sched}$ and the definition of the matrix $A$
in~(\ref{eq:sup-matrix-form}) we have
$\pmatrixsched{\sched}_{i,j} = A_{\tran^{i},s_{j}} = \tran^{i}(\state_{j}) - Q_{i,j}$,
where $Q_{i,j} = 1$ if $\tstate{\tran^{i}} = \state_{j}$, or otherwise $Q_{i,j} = 0$.
By Lemma~\ref{lemma:matrix-are-chosen}, we have
$\tstate{\tran^{i}} = \state_{j}$ iff $i=j$. Hence $Q_{i,j}$ is the identity matrix and
$\pmatrixsched{\sched}_{i,j} = \tran^{i}(\state_{j}) - I_{i,j} = \matrixsched{\sched}_{i,j} - I_{i,j}$, which completes the proof.
\end{proof}

The matrix $\pmatrixsched{\sched}$ happens to be very important in our proofs. We profit
from the fact that it is non-singular provided that $\sched$ is \apttext.

\begin{lemma}
\label{lemma:sched-mat-non-singular}
For all \apttext\ $\sched$, the matrix $\pmatrixsched{\sched}$ is non-singular.
\end{lemma}
\begin{proof}
Suppose, towards a contradiction, that there exists $\vec{x} \not= \vec{0}$
such that $\pmatrixsched{\sched}\vec{x} = \vec{0}$. Then, by Lemma~\ref{lemma:pmatrix-matrix},
we have $(\matrixsched{\sched} - I)\vec{x} = \vec{0}$, which implies
$\matrixsched{\sched}\vec{x} = \vec{x}$ and hence $(\matrixsched{\sched})^{z}\vec{x} = \vec{x}$
for all $z \geq 0$. We arrive to a contradiction by showing that for all $j$ there exists
$z$ such that
\begin{equation}
\label{eq:contr-sched-mat-non-singular}
 \abs{((\matrixsched{\sched})^{z}\vec{x})_{j}} < \max_{\pstate} \abs{x_{\pstate}} \; .
\end{equation}
In particular, for $q = \arg \max_{\pstate} \abs{x_{\pstate}}$ this yields
$\abs{((\matrixsched{\sched})^{z}\vec{x})_{q}} < \abs{x_{q}}$, which contradicts
$(\matrixsched{\sched})^{z}\vec{x} = \vec{x}$.

Now we prove~(\ref{eq:contr-sched-mat-non-singular}). Since $\sched$ is \apttext,
from every $\state_{j} \in \maybe$ there exists a path
$\pat \in \pathsched{\sched}{\state_{j}}$ with $\last{\pat} \in \tsts$, such that all
the states previous to $\last{\pat}$ are not in $\tsts$. We prove that
$z$ can be taken to be $\pathlen{\pat}$.
We proceed by induction on the length of
$\pat$. If $\pathlen{\pat} = 1$, by Lemma~\ref{lemma:matrix-are-chosen} we have
$\sched(\state_{j})(u) = \tran^{j}(u) > 0$ for some
$u \in \tsts$, and hence\footnote{The result for discounted \mdpstext\ does not
use $\maybe$ as the analogous of this sum is always less than $1$ due to the discounts}
$\sum_{\statep \in \maybe} \tran^{j}(\statep) < 1$.
Taking $z=1$ we obtain \[ \abs{(\matrixsched{\sched}\vec{x})_{j}}
            = \abs{\sum_{\statep \in \maybe} \tran^{j}(\statep) x_{\statep}}
			\leq \sum_{\statep \in \maybe} \tran^{j}(\statep) \abs{x_{\statep}}
			\leq \sum_{\statep \in \maybe} \tran^{j}(\statep)
			                       \max_{\pstate}  \abs{x_{\pstate}}
			< \max_{\pstate} \abs{x_{\pstate}} \; , \]
which proves that we can take $z=1=\pathlen{\pat}$. The last strict inequality holds only if
$\max_{\pstate} \abs{x_{\pstate}} > 0$, which follows from $\vec{x} \not= \vec{0}$.

If $\pathlen{\pat} = l+1$, there exists $\state_{q} \in \maybe$ such that
$\tran^{j}(\state_{q}) > 0$
and $q$ reaches $\tsts$ in $l$ steps. The inductive hypothesis holds for $q$, and
hence
$\abs{((\matrixsched{\sched})^{l}\vec{x})_{q}} < \max_{\pstate} \abs{x_{\pstate}}$,
from which we obtain:
\begin{multline*}
\abs{((\matrixsched{\sched})^{l+1}\vec{x})_{j}}
=
\abs{(\matrixsched{\sched}(\matrixsched{\sched})^{l}\vec{x})_{j}}
			\leq \sum_{\statep \in \maybe \setminus \{ s_{q} \}}
			       \tran^{j}(\statep) \: \abs{((\matrixsched{\sched})^{l}\vec{x})_{\statep}}
			       \; + \; \tran^{j}(s_{q}) \: \abs{((\matrixsched{\sched})^{l}\vec{x})_{q}}
			\\
			= \sum_{\statep \in \maybe \setminus \{ s_{q} \} } \tran^{j}(\statep)
			                       \: \abs{x_{\statep}}
			      \; + \; \tran^{j}(s_{q}) \: \abs{((\matrixsched{\sched})^{l}\vec{x})_{q}}
			\leq 
			      \sum_{\statep \in \maybe \setminus \{ s_{q} \} } \tran^{j}(\statep)
			                       \max_{\pstate}  \abs{x_{\pstate}}
			      \; + \; \tran^{j}(s_{q}) \: \abs{((\matrixsched{\sched})^{l}\vec{x})_{q}}
		    \\ < \sum_{\statep \in \maybe \setminus \{ s_{q} \}} \tran^{j}(\statep)
			                       \max_{\pstate}  \abs{x_{\pstate}}
			       \; + \; \tran^{j}(s_{q}) \max_{\pstate} \abs{x_{\pstate}}
			\leq \max_{\pstate} \abs{x_{\pstate}}
\end{multline*}
This finishes the proof of~(\ref{eq:contr-sched-mat-non-singular}). Assuming that 
$\pmatrixsched{\sched}\vec{x} = \vec{0}$ for some $\vec{x} \not= \vec{0}$,
we derived~(\ref{eq:contr-sched-mat-non-singular}), which contradicts $(\matrixsched{\sched})^{z}\vec{x} = \vec{x}$ for all $z \geq 0$, thus finishing
the proof.
\end{proof}

\begin{lemma}
\label{lemma:matrix-non-singular}
For all \apttext\ schedulers $\sched$, the basis defined in~(\ref{eq:def-basis})
is non-singular.
\end{lemma}
\begin{proof}
We show that the equation
$\bsched{\sched}\vec{x} = \vec{0}$ holds only if $\vec{x} = \vec{0}$.
Note that the vector $\vec{x}$ has one component for each column of the basis,
that is, one component for each state in $\maybe$ (called $x_{\state}$), and one component for
each transition such that $\sched(\tstate{\tran}) \not= \tran$ (called $x_{\tran}$).
The matrix equation $\bsched{\sched}\vec{x} = \vec{0}$ corresponds to $\alltran$ equations,
one for each transition. If $\tran \in \bsched{\sched}$, since $\statep \in \bsched{\sched}$
for all $\statep \in \maybe$, the equation corresponding to $\tran$ is
\begin{equation}
\label{eq:restrict-not-chosen}
 \sum_{\statep \in \maybe}
    	   A_{\tran,\statep} x_{\statep} + x_{\tran}  = 0 \; .
\end{equation}
If $\tran \not\in \bsched{\sched}$, the corresponding equation is
\begin{equation}
\label{eq:restrict-chosen}
 \sum_{\statep \in \maybe}
    	   A_{\tran,\statep} x_{\statep}  = 0 \; .
\end{equation}
(Note that the sum term $x_{\tran}$ has disappeared. This corresponds to the
fact that the column corresponding to $x_{\tran}$ is not in the basis.)
Since the transitions $\tran \not\in \bsched{\sched}$ are those such that
$\sched(\tstate{\tran}) = \tran$, the set of equations~(\ref{eq:restrict-chosen})
is equivalent to $\pmatrixsched{\sched}\vec{s} = \vec{0}$, where $\vec{s}$ is the
subvector of $\vec{x}$ having only the components corresponding to states. In consequence, if
$\bsched{\sched}\vec{x} = \vec{0}$ holds, then in particular
$\pmatrixsched{\sched}\vec{s} = \vec{0}$ and, since $\sched$ is \apttext, by
Lemma~\ref{lemma:sched-mat-non-singular} it must be $\vec{s} = \vec{0}$, that is,
$x_{\statep} = 0$ for all $\statep \in \maybe$. Using this
in~(\ref{eq:restrict-not-chosen}) we have $x_{\tran}  = 0$ for all
$\tran \in \bsched{\sched}$. We have proven $x_{j} = 0$ for every component
$j$ of $\vec{x}$, thus showing $\vec{x} = \vec{0}$.
\end{proof}

\begin{theorem}
\label{thm:optimal-is-feasible}
If a scheduler is optimal as in~\textup{(\ref{eq:optimal-sched-sup})}
(or~\textup{(\ref{eq:optimal-sched-inf}}), resp.) then the solution
induced by the basis $\bsched{\sched}$ is feasible.
\end{theorem}
\begin{proof}
Let $\vec{x}$ be the solution induced
by $\bsched{\sched}$ for some optimal $\sched$. By
Lemma~\ref{lemma:basic-satisfies-eq}, we need to prove
$\vec{x} \geq \vec{0}$. We prove this inequality by showing that
$x_{\state} = \prrhmi{\sched}{\tsts}{\mdp}{\state} \geq 0$ for all
$\state$ and $x_{\tran} \geq 0$ for all $\tran$.

Since in $\bsched{\sched}$ the variables $x_{\tran} \in \trans_{\sched}$
are non basic, in the solution $\vec{x}^{\sched}$ induced by
$\bsched{\sched}$ we have $x_{\tran}=0$ for
all $\tran \in \trans_{\sched}$. Then, using
Lemma~\ref{lemma:basic-satisfies-eq} for our particular constraint
matrix $A|I$, we obtain
\begin{equation}
\label{eq:values-sched}
 x_{\state} = \sum_{\statep \in \maybe} \sched(\state)(\statep)
      \: x_{\statep} \; + \; \sum_{\statep \in \tsts} \sched(\state)(\statep) \; .
\end{equation}
This is equivalent to $\pmatrixsched{\sched}\vec{x} = \vec{q}$ for
some vector $\vec{q}$. By Lemma~\ref{lemma:sched-mat-non-singular},
there exists exactly one $\vec{x}$ satisfying~(\ref{eq:values-sched}).
Let $v^{\sched}_{\state}$ be
$\prrhmi{\sched}{\tsts}{\mdp}{\state}$. 
A classic result for
\mdpstext\ (see, for
instance,~\cite[Section 4.2]{DBLP:conf/sfm/ForejtKNP11},
\cite[Theorem 3.10]{STAN//CS-TR-98-1601}) states that, for
an optimal scheduler $\sched$, it holds
\begin{equation}
\label{eq:sched-chooses-max}
 v^{\sched}_{\state} = \max_{\tran \in \enabledt{\state}}
         \sum_{\statep \in \maybe} \tran(\statep) v^{\sched}_{\statep} +  
                   \sum_{\statep \in \tsts} \tran(\statep)
\end{equation}
and
\[ \sched(\state) \in \arg \max_{\tran \in \enabledt{\state}}
        \sum_{\statep \in \maybe} \tran(\statep) v^{\sched}_{\statep}  +  
                   \sum_{\statep \in \tsts} \tran(\statep) \; . \]
for all states $\state$. From the last two equations:
\[ v^{\sched}_{\state} = 
     \sum_{\statep \in \maybe} \sched\!(\state)(\statep) \: v^{\sched}_{\statep}  \; +  
             \; \sum_{\statep \in \tsts} \sched\!(\state)(\statep) \; . \]
This is equivalent to
$\pmatrixsched{\sched}\vec{v^{\sched}} = \vec{q}$ as before.
After~(\ref{eq:values-sched}) we have seen that this equation has a
unique solution, and so $x_{\state} = v^{\sched}_{\state}$ for all
$\state \in \maybe$. By~(\ref{eq:sched-chooses-max}) we have
\begin{equation}
\label{eq:ineq-state}
 x_{\state} \geq \sum_{\statep \in \maybe} \tran(\statep) \:
       x_{\statep} \; + \; \sum_{\statep \in \tsts} \tran(\statep)
\end{equation}
for all $\state \in \maybe$, $\tran \in \enabledt{\state}$.
Applying Lemma~\ref{lemma:basic-satisfies-eq} to our particular
constraint matrix $A|I$, we have
\[ x_{\tran} = x_{\state} - \sum_{\statep \in \maybe} \tran(\statep)
       x_{\statep} - \sum_{\statep \in \tsts} \tran(\statep) \; . \]
Hence, $x_{\tran} \geq 0$ for all $\tran$ by~(\ref{eq:ineq-state}).
In conclusion,
$x_{\state} =  v^{\sched}_{\state} \geq 0$ for all $\state \in \maybe$
and $x_{\tran} \geq 0$ for all $\tran$. Then, the solution $\vec{x}$
induced by $\bsched{\sched}$ is feasible.

For the case of the minimum, the analogue of~(\ref{eq:sched-chooses-max})
is:
\begin{equation}
\label{eq:sched-chooses-min}
 v^{\sched}_{\state} = \min_{\tran \in \enabledt{\state}}
         \sum_{\statep \in \maybe} \tran(\statep) v^{\sched}_{\statep} +  
                   \sum_{\statep \in \tsts} \tran(\statep)
\end{equation}
The fact that the equation
$\pmatrixsched{\sched}\vec{v^{\sched}} = \vec{q}$ has a unique solution
again yields $x_{\state} = \vec{v^{\sched}}$. For $x_{\tran}$,
using the constraint matrix $-A|I$ for the minimum and~(\ref{eq:sched-chooses-min}) we obtain
\[ x_{\tran} = -x_{\state} + \sum_{\statep \in \maybe} \tran(\statep)
       x_{\statep} + \sum_{\statep \in \tsts} \tran(\statep)
       = \sum_{\statep \in \tsts} \tran(\statep) + \sum_{\statep \in \maybe} \tran(\statep)
        - x_{\statep} \geq 0 \; . \]
\end{proof}

\begin{theorem}
\label{thm:apt-is-dual-feasible}
Given an \apttext\ scheduler $\sched$, the solution induced by the basis $\bsched{\sched}$
is dual feasible. (For the definition of \emph{dual feasible} see
Subsection~\ref{subsec:lin-prog}.)
\end{theorem}
\begin{proof}
First we find a matrix expression for $\bsched{\sched}^{-1}$.
Suppose we reorder the rows of $\bsched{\sched}$ so that the rows corresponding to
transitions in the basis occur first. The resulting matrix is
\[
 \basis'_{\sched} = \left(
\begin{array}{c|c}
A' & I^{(\alltran-\allstates) \times (\alltran-\allstates)}  \\ \hline
\pmatrixsched{\sched} & 0
\end{array}
\right)
\]
where $A'$ is a submatrix of $\bsched{\sched}$.
We can write
\begin{equation}
\label{eq:reordered-basis}
 \basis'_{\sched} = P \bsched{\sched}
\end{equation}
where $P$ is a permutation matrix.
In order to find the inverse of $\basis'_{\sched}$ we pose the following matrix equation:
\begin{multline*} 
\left(\begin{array}{c|c}
A' & I^{(\alltran-\allstates) \times (\alltran-\allstates)}  \\ \hline
\pmatrixsched{\sched} & 0
\end{array}\right)
\left(\begin{array}{c|c}
A_{11} & A_{12} \\ \hline
A_{21} & A_{22}
\end{array}\right)
\\ = 
\left(\begin{array}{c|c}
A' A_{11} + A_{21} & A' A_{12} + A_{22} \\ \hline
\pmatrixsched{\sched} A_{11} & \pmatrixsched{\sched} A_{12}
\end{array}\right)
= I = 
\left(\begin{array}{c|c}
I^{\allstates \times \allstates} & 0 \\ \hline
0 & I^{(\alltran-\allstates) \times (\alltran-\allstates)}
\end{array}\right)
\end{multline*}
These equations,
suggest that we can take $A_{11} = 0$, and hence $A_{21} = I$. Moreover,
it must be $A_{12} = \pmatrixsched{\sched}^{-1}$
(which exists by Lemma~\ref{lemma:matrix-non-singular}) and hence
$A_{22} = -A'\pmatrixsched{\sched}^{-1}$.
The equation below can be easily checked by verifying that
$\basis'^{-1}_{\sched} \basis'_{\sched} = I$
\begin{equation}
\label{eq:basis-inverse}
 \basis'^{-1}_{\sched} = \left( \begin{array}{c|c}
     0^{\allstates \times (\alltran-\allstates)} & \pmatrixsched{\sched}^{-1} \\
                                                                                    \hline
     I^{(\alltran-\allstates) \times (\alltran-\allstates)} & -A' \pmatrixsched{\sched}^{-1}
     \end{array} \right)
\end{equation}

Next, we use~(\ref{eq:basis-inverse}) to show that the reduced costs depend only on the
constraint coefficients of the transitions chosen by the scheduler.

We consider first the case of the maximum. Recall that our constraint matrix
is $A|I$ and the costs $c_{\tran}$ associated to the transitions variables are $0$
for all $\tran$ (see~(\ref{eq:sup-matrix-form})).
According to the definition of reduced cost (see Subsection~\ref{subsec:lin-prog}), to
prove dual feasibility we need to show
$- \vec{c}^{\bsched{\sched}} \bsched{\sched}^{-1} I_{\tran} \geq 0$
for all $\tran \not\in \bsched{\sched}$, where $I_{\tran}$ is the column of the identity matrix
corresponding to $\tran$. From~(\ref{eq:reordered-basis}), we have
$\bsched{\sched}^{-1} = \basis'^{-1}_{\sched} P$, and hence
our inequality is $- \vec{c}^{\bsched{\sched}} \basis'^{-1}_{\sched} P I_{\tran} \geq 0$. 
Since $P$ is a permutation matrix, we know that $P I_{\tran}$ is a column of the
identity matrix, say $I_{k(\tran)}$. Given our costs in~(\ref{eq:sup-matrix-form}), and given
the definition of $\bsched{\sched}$, we have that $\vec{c}^{\bsched{\sched}}$ is the vector
$(\overbrace{1,\cdots,1}^{\allstates},
                          \overbrace{0,\cdots,0}^{\alltran-\allstates})$, and hence
from~(\ref{eq:basis-inverse}) we get
$\vec{c}^{\bsched{\sched}} \basis'^{-1}_{\sched} =
       (\vec{0}^{1 \times (\alltran-\allstates)}, \: \vec{1}^{1 \times \allstates} \pmatrixsched{\sched}^{-1})$.
In conclusion, we have proven
\begin{equation}
\label{eq:dual-feasible-ineq}
- \vec{c}^{\bsched{\sched}} \bsched{\sched}^{-1} I_{\tran} =
 -(\vec{0}^{1 \times (\alltran-\allstates)}, \: \vec{1}^{1 \times \allstates}
                     \pmatrixsched{\sched}^{-1}) \; I_{k(\tran)} \; ,
\end{equation}
and we must prove that this number is greater than or equal to $0$ for all
$\tran \not\in \bsched{\sched}$.

Whenever $k(\tran) \leq \alltran-\allstates$, the result holds
since~(\ref{eq:dual-feasible-ineq}) is $0$.

In case $k(\tran) > \alltran-\allstates$, we prove the result using the fact that these
values depend only on the transitions chosen by $\sched$. In fact, given the \mdptext\
$\mdp$ and the scheduler $\sched$, if we write~(\ref{eq:dual-feasible-ineq})
for the Markov chain $\srestrict{\mdp}{\sched}$ (see Def.~\ref{def:restrict-mc}),
we obtain
\begin{equation}
\label{eq:dual-feasible-ineq-rest}
 - \vec{c}^{\bsched{(\srestrict{\mdp}{\sched}), \sched}} \; \bsched{(\srestrict{\mdp}{\sched}),\sched }^{-1} \; I_{\tran}
 \; \: = \: \: -\vec{1}^{1 \times \allstates} \pmatrixsched{\sched}^{-1} \; I_{\tran}
\end{equation}
for all $\tran \not\in \bsched{(\srestrict{\mdp}{\sched}), \sched}$.
Note that for $\srestrict{\mdp}{\sched}$ there is no need to reorder (as there are no
transitions in the basis) and so $\tran = k(\tran)$. Given that all the transitions
$\srestrict{\mdp}{\sched}$ are chosen by $\sched$, the basis
$\bsched{(\srestrict{\mdp}{\sched}), \sched}$ contains all the states and no transitions.
In this equation, $I_{\tran}$ can be \emph{any} column of $I^{\allstates \times \allstates}$
(again, due to the fact that there are no transitions in the basis).

Suppose, towards a contradiction, that~(\ref{eq:dual-feasible-ineq}) is less than $0$
for some $k(\tran) > \alltran-\allstates$. This is equivalent to
$-\vec{1}^{1 \times \allstates}
              \pmatrixsched{\sched}^{-1} \; I_{k(\tran) - (\alltran-\allstates)} < 0$.
By~(\ref{eq:dual-feasible-ineq-rest}) we have
$-\vec{1}^{1 \times \allstates} \pmatrixsched{\sched}^{-1} \; I_{\tran'} < 0$ for
some $\tran'$ in $\srestrict{\mdp}{\sched}$. Then, the solution
induced by the basis is not dual feasible for the problem associated to
$\srestrict{\mdp}{\sched}$. As there is at least one optimal basic and dual feasible
solution (the one found by simplex method), there exists an optimal solution $\vec{x}^{C}$
such that the corresponding basis $B^{C}$ is not
$\bsched{(\srestrict{\mdp}{\sched}), \sched}$. As in $\srestrict{\mdp}{\sched}$
there exists only one basis containing all states (namely
$\bsched{(\srestrict{\mdp}{\sched}), \sched}$), there exists
$\state \not\in B^{C}$. In consequence, we have
$x^{C}_{\state} = 0$. Since $x^{C}$ is optimal, by Theorem~\ref{thm:lp-problems}, we obtain
$\prrhmi{\sched}{\tsts}{\srestrict{\mdp}{\sched}}{\state} = 0$, from
which~(\ref{eq:rest-mc-eq-probs}) yields
$\prrhmi{\sched}{\tsts}{\mdp}{\state} = 0$. This contradicts the fact that $\sched$
is \apttext.

The proof for the case of the minimum is completely analogous: despite
the differences in the constraints and the cost vector, the
reduced costs in~(\ref{eq:dual-feasible-ineq}) are the same as before:
\[ - \vec{c}^{\bsched{\sched}} \bsched{\sched}^{-1} I_{\tran} =
 -(-\vec{0}^{1 \times (\alltran-\allstates)}, \: -\vec{1}^{1 \times \allstates}
                     (- \pmatrixsched{\sched}^{-1})) \; I_{k(\tran)}
   = - (\vec{0}^{1 \times (\alltran-\allstates)}, \: \vec{1}^{1 \times \allstates}
                     \pmatrixsched{\sched}^{-1}) \; I_{k(\tran)} \; . \]
These values again coincide with the ones in
a system having only the transitions chosen by $\sched$.
\end{proof}


\begin{proof}[Proof (of Theorem~\ref{thm:alg-correct})]
If the algorithm returns a value, then it is
$\arg \min_{\vec{x}} \lpproblem$, where $\lpproblem$ is the
problem~(\ref{eq:sup-matrix-form}) (or~(\ref{eq:inf-matrix-form})
for the minimum). Hence, by Theorem~\ref{thm:lp-problems}, the returned
value coincides with the \textbf{output} specification.
We have that if $\sched$ is \apttext, then $\bsched{\sched}$ is a basis
by Lemma~\ref{lemma:matrix-non-singular}. As a consequence, the algorithm
never enters the branch in line~\ref{line:basis-singular}, and so
the result is returned.
The termination in a single iteration is a consequence of the
fact that the solution corresponding to an optimal scheduler
$\sched$ is both feasible (Theorem~\ref{thm:optimal-is-feasible}) and
dual feasible (Theorem~\ref{thm:apt-is-dual-feasible}).
\end{proof}

\section{Experimental results}
\label{sec:exp-results}

\noindent\textbf{Implementation.}~We implemented our method by extending the model
checker \prism~\cite{KNP11}, using the \lptext\ library \glpk~\cite{glpk}.
We compiled \glpk\ using the library for arbitrary precision \texttt{gmp}.
We needed to modify
the code of \glpk: although there is a solver function that uses exact
arithmetic internally, this function does not allow us to retrieve the exact
value. Aside from these changes to \glpk\ and some additional code scattered
around the \prism\ code (in order to gather information about the
scheduler), the specific code for implementing our method is less than 300
lines long. With these modifications, \prismtext\ is able to print the
numerator and the denominator of the probabilities calculated.

Our implementation works as follows: in the first step, we use the value iteration
already implemented in \prismtext\ to calculate a candidate scheduler. In the
next step, the \lptext\ problem is constructed by iterating over each state: for
each transition enabled, the corresponding probabilities are inserted in the matrix.
The basis is constructed along this process: when a transition is considered, the
description of the scheduler (implemented as an array) is queried about whether this
transition is the one chosen
by the scheduler. Next we solve the \lptext\ problem. For the reasons explained in
Subsection~\ref{subsec:lin-prog}, in Algorithm~\ref{alg:main-algorithm} we use
the dual simplex method, except when we compare it to the primal one. The
reader familiar with \texttt{glpk} might notice that the dual variant is not
implemented under exact arithmetic on \texttt{glpk}: to overcome this, instead of
providing \texttt{glpk} with the original problem, we provided the dual problem and
retrieved the values of the dual variables (the dual problem is obtained by providing
the transpose of the constraint matrix and by negating the cost coefficients, and so
it does not affect the running time).

The experiments were carried out on an Intel i7 @3.40Ghz with 8Gb RAM,
running Windows~7.

\vspace{0.2cm}\noindent\textbf{Case studies.}~We studied three known models
available from the \prism\ benchmark suite~\cite{prism-cases}, where the
reader can look for matters not explained here (for instance, details about
the parameters of each model). For the parameters whose values are not
specified here, we use the default values. In the
IEEE 802.11 Wireless LAN model, two stations use a randomised exponential 
backoff rule to minimise the likelihood of transmission collision.
The parameter $N$ is the number of maximum backoffs. We compute the
maximum probability that the backoff counters of both stations reach
their maximum value.
The second model concerns the consensus algorithm for $N$ processes of
Aspnes \& Herlihy~\cite{DBLP:journals/jal/AspnesH90}, which
uses shared coins. We calculate the maximum probability that the
protocol finishes without an agreement. The parameter
$K$ is used to bound a shared counter. Our third case study
is the IEEE 1394 FireWire Root Contention Protocol
(using the \prism\ model which is based on~\cite{SV99}). We
calculate the minimum probability that a leader is elected
before a deadline of $D$ time units. 


\vspace{0.2cm}\noindent\textbf{Linear programming versus Algorithm~\ref{alg:main-algorithm}.}
Table~\ref{tab:exp-results} allows us to compare (primal and dual) simplex starting from
a default basis, against Algorithm~\ref{alg:main-algorithm}, which
provides a starting basis from a candidate scheduler. 
Aside from the construction of the \mdptext\ from the \prismtext\ language
description (which is the same either using \lptext\
or Algorithm~\ref{alg:main-algorithm}, and is thus disregarded in our comparisons),
the steps in our implementation are:
\begin{inparaenum}[(1)]
\item{perform value iteration to obtain a candidate scheduler;}
\item{construct the \lptext\ problem;}
\item{solve the problem in exact arithmetic in zero or more iterations (the
latter is the case in which the scheduler is not optimal).}
\end{inparaenum}
All these times are shown in Table~\ref{tab:exp-results}, as well as its sum,
expressed in seconds.
The experiments for \lptext\ were run with a time-out of one hour (represented
with a dash).

\begin{table}
{ \centering \scriptsize
\begin{tabular}{|l|r|r|r||r||r||r|r|r||r||} \hline
 & & & & \multicolumn{6}{c|}{Time (seconds)} \\ \cline{5-10}
 & & & & \multicolumn{2}{c||}{\lptext\ without Alg.~\ref{alg:main-algorithm}} & \multicolumn{4}{c||}{Algorithm~\ref{alg:main-algorithm}} \\ \cline{5-10}
\parbox{1.3cm}{Model} & \parbox{0.7cm}{Para\-meters} & \parbox{1cm}{$\allstates = \card{\maybe}$} & 
\parbox{0.7cm}{$\alltran$} & \parbox{0.8cm}{Primal} & \parbox{0.8cm}{Dual} &
\parbox{0.8cm}{Value iter.} & \parbox{1.1cm}{LP constr.} & \parbox{0.9cm}{Dual simplex}& \parbox{1.3cm}{Total}
\\ \hline
Wlan & 3 & 2529 & 96302 & 19.53 & 11.76 & 0.36 & 0.05 & 0.03 & 0.44 \\
(N) & 4 & 5781 & 345000 & 110.32 & 61.83 & 2.30 & 0.21 & 0.06 & 2.57 \\
 & 5 & 12309 & 1295218 & 535.76 & 326.64 & 14.93 & 1.32 & 0.15 & 16.40 \\ \hline
Consensus  & 3,3 & 3607 & 3968 & 251.74 & 35.32 & 2.93 & 0.04 & 0.15 & 3.12 \\
(N, K) & 3,4 & 4783 & 5216 & 488.84 & 64.00 & 6.47 & 0.06 & 0.58 & 7.11 \\
 & 3,5 & 5959 & 6464 & 1085.70 & 105.36 & 12.74 & 0.06 & 1.87 & 14.67 \\
 & 4,1 & 11450 & 12416 & - & 432.98 & 2.88 & 0.11 & 0.19 & 3.18 \\
 & 4,2 & 21690 & 22656 & - & 1951.91 & 20.41 & 0.23 & 0.37 & 21.01 \\
 & 4,3 & 31930 & 32896 & - & - & 59.73 & 0.49 & 0.58 & 60.80 \\
 & 4,4 & 42170 & 43136 & - & - & 134.62 & 0.64 & 0.78 & 136.04 \\
 & 4,5 & 52410 & 53376 & - & - & 246.90 & 0.91 & 0.96 & 248.77  \\ \hline
Firewire & 200 & 1071 & 80980 & 4.50 & 2.65 & 0.28 & 0.04 & 0.01 & 0.33 \\
(D)  & 300 & 23782 & 213805 & - & 1314.32 & 2.89 & 1.04 & 0.24 & 4.17 \\
 & 400 & 81943 & 434364 & - & - & 11.05 & 8.74 & 0.88 & 20.67  \\ \hline
\end{tabular}

}
\caption{Comparison of primal and dual simplex starting from a default basis against
                 Algorithm~\ref{alg:main-algorithm}\label{tab:exp-results}}
\end{table}

Our method always outperforms the naive application of \lptext. The case
with the lowest advantage is Consensus~(3,5), and still our method takes less
than $1/6$ of the time required by dual simplex. 

With respect to the time devoted to exact arithmetic in
Algorithm~\ref{alg:main-algorithm},
in all cases the simplex under exact arithmetic takes a fraction
of the time spent by the other operations of the algorithm (namely, to perform
value iteration and to construct the \lptext\ problem). In Consensus~(3,5),
the simplex algorithm takes less than $1/6$ of the time devoted to the other
operations. In all other cases the ratio is even lower.

The greatest number found was $28821938103543398400$, the denominator
in the solution of Firewire 400. It needs 65 bits to be stored. The
computations were performed using 32 bit libraries, and so
the exact arithmetic computations used around 3 words in the worst case
(which is not really a challenge for an arbitrary precision library).
We can conclude that, even for systems with more than $10000$
states (up to $80000$, in our experiments), the overhead introduced
by exact arithmetic is manageable.

\vspace{0.2cm}\noindent\textbf{Suboptimal schedulers as suboptimal bases.}
Other than measuring whether the calculation is reasonably quick in case the
scheduler from \prismtext\ is optimal, a secondary measuring concerns how close
is the basis to an optimal one in case the scheduler provided by \prismtext\ is
\emph{not} optimal.

Except in cases Consensus~(3,$\cdot$), simplex stopped after $0$
iterations, thus indicating that \prismtext\ was able to find the optimal
scheduler. For optimal schedulers there is no difference
between using primal or dual simplex in Algorithm~\ref{alg:main-algorithm}
(we ran the experiments and the running time
of the simplex variants differed by at most $0.05$ seconds).

The probabilities obtained in each step of the value iteration converge to those
of an optimal scheduler. Given a threshold $\epsilon$, value iteration stops only
after $\abs{x_{s} - x'_{s}} \leq \epsilon$ for all $s$, where $\vec{x}$ and
$\vec{x}'$ are the vectors obtained in the last two iterations.

In Table~\ref{tab:time-basis} we compare the amount of iterations and the time
spent by primal and dual simplex for schedulers obtained using
different thresholds. We considered only the cases Consensus~(3,$\cdot$),
as in other cases the scheduler returned by \prismtext\ was optimum
except for gross thresholds above $0.05$, which are rarely used in practice
(the default $\epsilon$ in \prismtext\ is $10^{-6}$).
In addition to the default value, we considered representatives
the value $10^{-7}$ (since $10^{-8}$
already yields the exact solution for $(3,3)$ in the dual case: a value smaller than
$10^{-7}$ would have yielded uninteresting numbers for this case) and the value
$10^{-16}$, since in $(3,5)$ the scheduler does not improve beyond such threshold.
In fact, for $10^{-16}$ the result is the same as for $10^{-323}$, and $10^{-324}$
is not a valid \texttt{double}. In Java, the type \texttt{double} corresponds to a
IEEE~754 64-bit floating point.

In consequence, we have one case (namely, Consensus~(3,5)), where \prismtext\
cannot find the worst-case scheduler for any \texttt{double} threshold
(and thus should be recoded to use another arithmetic primitives to get
exact results), while our method is able to calculate exact results using
less than two seconds after value iteration, as shown in Table~\ref{tab:exp-results}.

For Consensus~(3,$\cdot$) we see that dual simplex performs betters
than primal simplex. Consensus~(3,4) shows that the primal simplex can behave worse
when starting from $\bsched{\sched}$ than the dual simplex starting from the default
basis (compare with the corresponding row in Table~\ref{tab:exp-results}).
Moreover, it can be the case that it takes \emph{more} time as the threshold
decreases (note that, in contrast, in Consensus~(3,5) the time decreases with the
threshold, as expected). This suggests that the dual variant should be preferred over the primal.

Comparing against Table~\ref{tab:exp-results}, we see that, for each variant
of the simplex method, starting from the basis $\bsched{\sched}$ results in
a quicker calculation than starting from the default basis.

\begin{table}
{ \centering \footnotesize
\begin{tabular}{|l||r|r|r|r|r|r|r|r|r|r|r|r|} \hline
   &  \multicolumn{6}{|c}{\parbox{1.5cm}{Primal}}  &  \multicolumn{6}{|c|}{\parbox{1.5cm}{Dual}}  \\ \cline{2-13}
   & \multicolumn{3}{|c}{\parbox{1.5cm}{Iterations}} & \multicolumn{3}{|c}{Time (seconds)} & \multicolumn{3}{|c}{Iterations} & \multicolumn{3}{|c|}{Time (seconds)} \\ \hline
$\epsilon$ ($10^{-n}$) & 6 & 7 & 16 & 6 & 7 & 16 & 6 & 7 & 16 & 6 & 7 & 16  \\ \hline
Consensus (3,3) & 187 & 134 & 0 & 3.43 & 2.43 & 0.10 & 10 & 6  & 0 & 0.19 & 0.15 & 0.10 \\
Consensus (3,4) & 2497 & 6278 & 0 & 74.42 & 202.58 & 0.14 & 37 & 28 & 0 & 0.63 & 0.51 & 0.13 \\
Consensus (3,5) & 4990 & 4340 & 1239 & 190.53 & 160.44 & 49.487 & 94 & 61 & 6 & 1.93 & 1.24 & 0.25 \\ \hline
\end{tabular}

}
\caption{Time spent when the starting basis is not optimal\label{tab:time-basis}}
\end{table}
                                                                                                                                                                           It is worth mentioning that in all cases the difference between the probabilities provided by \prismtext\ and
the exact values was less than the threshold for value iteration. It indicates that the finite precision of
the computations does not affect the results significantly.                                                                                                                                                
\section{Discussion and further work}

\noindent\textbf{Linear programming versus policy iteration.}~It
is known that the dual simplex method applied for discounted \mdpstext\ 
is just the same as \emph{policy iteration}
(for an introduction to this method see~\cite{DBLP:conf/sfm/ForejtKNP11}) seen from
a different perspective. Indeed, this has been used to obtain complexity bounds
(see~\cite{simplex-strongly-polynomial}). Theorems~\ref{thm:optimal-is-feasible}
and~\ref{thm:apt-is-dual-feasible} establish for undiscounted \mdpstext\ the same
correspondence between basis and schedulers as known for the discounted case, and
as a consequence the dual simplex is policy iteration disguised, also in the
undiscounted case.

Even without considering the results in this paper at all, exact solutions can also
be calculated by implementing policy iteration with exact arithmetic as, in each
iteration, the method calculates the probabilities corresponding to a scheduler and
checks whether they can be improved by another scheduler. Roughly speaking, if the
calculation and the check are performed using exact arithmetic, then the result is
also exact.

Despite this existing alternative, the correspondence between bases and schedulers we
presented in this paper allows to obtain an exact solution by using \lptext\ solvers,
thus profiting from all the knowledge concerning \lptext\ problems (and from existing
implementations such as \texttt{glpk}).

\vspace{0.2cm}\noindent\textbf{Complexity.}~To the best of our knowledge, the precise complexity
of the simplex method in our case is unknown. There are recent results for the
simplex  applied to similar problems. For instance, in~\cite{simplex-strongly-polynomial} it is proven that simplex is strongly polynomial for discounted
\mdpstext. Nevertheless,~\cite{DBLP:conf/icalp/Fearnley10}
shows an exponential lower bound to calculate rewards in the undiscounted case.
Unfortunately
(or not, as there is still hope that we can prove the time to be polynomial in our
case), the construction used in~\cite{DBLP:conf/icalp/Fearnley10} cannot be carried
out easily to our setting, as some of the rewards in the construction are negative
(and the equivalent to the rewards in our setting are the sums
$\sum_{\statep \in \tsts} \tran(\statep)$).

\vspace{0.2cm}\noindent\textbf{Further work.}~In the comparison of our method against \lptext, we
considered only the simplex method, as \glpk\ only implements this method in exact
arithmetic. The feasibility/applicability of other algorithms to solve \lptext\
problems using exact arithmetic is yet to be studied.

The fact that the probabilities obtained are exact allows
to prove additional facts about the system under consideration. For instance, the
exact values can be used in correctness certificates, or be the input of automatic
theorem provers, if they require exact values to prove some other properties of the
system. We plan to concentrate on these uses of exact probabilities.

\vspace{0.2cm}\noindent\textbf{Acknowledgements.}~The author is grateful to David Parker, Vojtech Forejt and
Marta Kwiatkowska for useful comments and proofreading.

\bibliographystyle{eptcs}
\bibliography{bibliography}

\end{document}